\documentclass[12pt]{article} 
\usepackage[margin=1in]{geometry} 
\usepackage{setspace}             

\usepackage[T1]{fontenc}  
\usepackage[utf8]{inputenc} 
\usepackage{lmodern} 
\usepackage{helvet}         
\usepackage{courier}        
\usepackage{type1cm}        
\usepackage{amssymb}        
\usepackage{amsmath}
\usepackage{bm}
\usepackage{makeidx}        
\usepackage{graphicx}       %
\usepackage{xcolor}
\usepackage{float}    
\usepackage{booktabs} 
\usepackage{array}
\usepackage{pdflscape}
\usepackage{multirow}
\usepackage{multicol}       
\usepackage[bottom]{footmisc}
\usepackage{algorithm}
\usepackage{algpseudocode}
\usepackage[numbers]{natbib} 
\usepackage{hyperref}
\usepackage{amsthm}
\newtheorem{theorem}{Theorem}[section]

\newcommand{\pp}{P}
\newcommand{\var}[3]{#1_{#2}(#3)}
\newcommand{\tx}{s}

\newcommand{\vX}[1]{\var{\boldsymbol{X}}{#1}{\tx}}
\newcommand{\ty}{t}

\newcommand{\Yv}{\var{Y}{\nu_1}{\ty},\dots,\; \var{Y}{\nu_N}{\ty}}
\newcommand{\Xv}[1]{\var{X}{\nu_1,#1}{\tx},\dots,\; \var{X}{\nu_N,#1}{\tx}}
\newcommand{\vPs}[1]{\var{\psi}{#1}{\ty}}
\newcommand{\vPh}[1]{\var{\phi}{#1}{\tx}}

\usepackage{authblk}

\begin{document}
	
	\title{Developments in Functional regression model for network structured data}
	\author[1]{Elvira Romano\thanks{\texttt{elvira.romano@unicampania.it}}}
	
	\author[1]{Antonio Irpino\thanks{\texttt{antonio.irpino@unicampania.it}}}
	\author[2]{Claire Miller\thanks{\texttt{Claire.Miller@glasgow.ac.uk}}}
	\affil[1]{Department of Mathematics and Physics, University of Campania ``Luigi Vanvitelli'' }
	\affil[2]{School of Mathematics and Statistics, University of Glasgow }
	
	\date{}
	\maketitle
	
	\abstract{
		In this  paper,  we propose a Network-Weighted Functional Regression (NWFR) model, an extension of Spatially Weighted Functional Regression (SWFR) to functional data defined on network-structured settings.
		To asses predictive uncertainity, we develop a functional conformal prediction procedure that yields a distribution free prediction intervals with guaranteed coverage.
		Through extensive evaluation on both simulated and real-world datasets, we demonstrate that the explicit modeling of network structure yields substantive improvements in point-prediction accuracy and markedly enhances the validity and precision of the resulting prediction intervals.
		
	}
	
	{{\bf keywords:} Conformal prediction, Functional data, Network data,  Functional Regression model}
	
	\section{Introduction}
	\label{sec:01} 
	Modern research frequently includes collecting numerous functional observations from various designs with complex dependencies, such as multivariate, longitudinal, spatial, networked, or time series data. These datasets, referred to as second-generation functional data  \cite{Koner:23}, go beyond conventional functional data by integrating intricate relationships between observations, departing from the independence assumption.
	This paper focuses on on appropriately modelling such interconnections, particularly within environments governed by network dependencies, to enable prediction and inference for continuous signals across nodes in a network that accounts for the dependencies induced by connections between nodes. This challenge lies at the intersection of Network Data Analysis \cite{ DeStefano:23, Ward:11} and Functional Data Analysis (FDA) \cite{Ramsay:05}, a field with significant foundational work in functional regression by Faraway \cite{Faraway:97}, Cardot et al. \cite{Cardot:99}, and others \cite{Aneiros-Perez:06, Chiou:04, Febrero-Bande:13, Febrero-Bande:10}.  
	In recent years, various methods have been proposed to model complex dependencies, such as multivariate functional data \cite{Morris:2015} and spatially dependent functional data \cite{Arnone:19, ROMANO44, Sangalli:13, Yamanishi:03}. However, studying interconnected functional entities within a network remains relatively novel \cite{Fontanella:20}.
	We introduce a regression model that utilizes the concept of representing multidimensional time series signals as functions observed at network nodes while incorporating edge weights to capture relationships between nodes. 
	In this context, the network itself represents the measurement process, while the measurements are continuous signals that exist over a spatial domain external to the network. The model integrates network weights to develop localized regression models, where each node’s estimation is influenced by its neighbors. This graph-based method allows for predictions without needing assumptions about the data distribution. To improve the dependability of the forecasts, we apply conformal prediction techniques, offering uncertainty estimates with guaranteed coverage.
	The paper is organized as follows.
	Section \ref{sec:02} defines  Functional Data on Networks.
	Section \ref{sec:03} formalizes the Functional regression model on the network.
	Section \ref{sec:04} presents an inferential framework for defining prediction intervals.
	Sections \ref{sec:05} and \ref{sec:06} detail the results of simulated and real-world environmental case studies.

	\section{Functional Data on Networks}	
	\label{sec:02}
	Communication network data provides a detailed description of how individuals or devices connect and interact within a network. It encompasses critical information such as the number of connections, their strength, and the communication patterns between nodes. 
	
	For example, when we consider the signal information over time as functional data, and consider that communication network data can be represented as a graph, we are dealing with functional data on a network.
	Let  $G = (\mathcal{V}, \mathcal{E})$  be an undirected graph where  $\mathcal{V}$ is the set of $N=|\mathcal{V}|$ vertices (or nodes) in the graph, and  $\mathcal{E} \subseteq \mathcal{V} \times \mathcal{V}$ is the set of edges (or links) connecting pairs of vertices.
	Each vertex  $\nu_i \in \mathcal{V}$ is described by a set of functional data $\vX{\nu_i} = \left( X_{\nu_i,1}(t), X_{\nu_i,2}(t), \dots, X_{\nu_i,\pp}(t) \right)$, where  $t \in T = [a,b]$  represents the time or index domain of the functional data, and $X_{\nu_i,p}(t)$ is a scalar function at time $t$ for the $p$-th component of the data in the vertex $\nu_i$. 
	
	Given the graph $G = (\mathcal{V}, \mathcal{E})$, the functional data at each vertex may be influenced by the data in the neighboring vertices, where the relationship between two vertices in the graph can be captured by their geodesic. For this purpose, we define the geodesic between a pair of vertices $\nu_i$ and $\nu_j$ as the length of the shortest path (or paths) connecting them. In the case of weighted graphs, the shortest path between two vertices is determined by minimizing a function of the edge weights along the connecting path.
	
	If the weights represent a connection, usually a sum of the weights of the edges is considered. If the graph is unweighted, this corresponds to the path, or the paths, with the fewest number of edges connecting $\nu_i$ to $\nu_j$.

	Let $\pi_{i,j}^\ell = \{(\nu_0,\nu_1), (\nu_1,\nu_2), \dots, (\nu_{\ell-1},\nu_\ell)\}$, where $\nu_0=\nu_i$ and $\nu_\ell=\nu_j$, $\ell$ is lesser or equal to the diameter of the network, represents a path of length $\ell$ from $\nu_i$ to $\nu_j$, such that each  $(\nu_h, \nu_{h+1})=e_{h,h+1} \in \mathcal{E}$ is an edge.
	Each edge $e_{h,k}\in\mathcal{E}$ has associated a weight $\omega(h,k)\geq0$. If $\omega(h,k)\geq0$ is considered as a cost measure, the total weight associated with $\pi_{i,j}^\ell$ is given as follows
	$$\omega\left(\pi_{i,j}^\ell\right)=\sum_{h=1}^\ell\omega(\nu_{h-1},\nu_h). $$
	

	The shortest path from $\nu_i$ to $\nu_j$ is the path (or the paths) $\pi_{ij}^*$ that minimizes the total weight among all possible paths from $\nu_i$ to $\nu_j$. Formally, if ${\Pi}_{ij}$ denotes the set of all paths from $\nu_i$ to $\nu_j$, then the shortest path is defined as
	
	$$
	\pi_{ij}^* = \arg \min_{\pi \in {\Pi}_{ij}} \sum_{e_{h,k} \in \pi} \omega(\nu_h, \nu_k),
	$$
	
	and the corresponding distance or geodesic (or minimum weight) is
	
	\begin{equation}\label{shortpath}
		d(i, j) = \min_{\pi \in {\Pi}_{ij}} \sum_{e_{h,k} \in \pi} \omega(\nu_h, \nu_k).    
	\end{equation}
	
	On such definition of geodesic distance depends the functional regression models on graphs, where the functional data at each vertex may be influenced by the data at neighboring vertices, and these dependencies are controlled by the graph structure.

	\section{ Functional Regression Model for network structured data}
	\label{sec:03}
	
	The classical functional concurrent regression model, as described in \cite{Ramsay:05}, is used when the response variable $Y(t)$ is a function of a set of $P$ functional covariates ${X_1(s), X_2(s), \dots, X_P(s)}$, where $t\in \mathcal{T}$, $s\in \mathcal{S}_p$, and $\mathcal{T}$ and $\mathcal{S}_p$, where $p=1,\ldots,P$, represent the domain over which the functions are observed. In this model, the goal is to model the relationship between the response variable $Y(t)$ and the set of functional predictors $X_1(s), X_2(s), \dots, X_P(s)$ . The general form of the classical functional regression model is given by:
	
	\begin{equation}
		Y(t) = \beta_0(t) + \sum_{p=1}^P \int_{\mathcal{S}_p} X_p(s) \beta_p(t, s) \, ds + \epsilon(t), \quad \forall t \in \mathcal{T},
	\end{equation}
	
	where $Y(t)$ is the response variable at time $t$,  $\beta_0(t)$ is the intercept function, $X_p(s)$ are the functional covariates observed over the domain  $\mathcal{S}_p$, $\beta_p(t, s)$ are the coefficient functions corresponding to each functional covariate  $X_p(s)$, and $\epsilon(t)$ represents the error term, typically assumed to be a random noise process.
	
	This model assumes that the observations $Y(t)$ are related to the functional predictors $X_p(s)$ through a linear combination of the coefficient functions $\beta_p(t, s)$, integrated over the domain $\mathcal{S}_p$. The model we propose extends this model to situations where the data are observed over a network structure, such as a sensor network or spatial locations connected by edges. In this framework, both the response variable and the functional covariates are observed at vertices (locations) of the network, and the relationships between vertices are captured by a weight matrix that incorporates the network structure into the regression analysis.
	
	Extending the GWFR \cite{Yamanishi:03} to network data, the NWFR model incorporates the network structure by introducing a weight matrix $ \mathbf{W}$, a similarity matrix where usually $w(i,j)\in[0,1]$ and is symmetric for undirected networks, which represents the proximity between vertices $ \nu_i $ and $\nu_j$. The elements of $\mathbf{W}$ depend on the geodesics and reflect the connectivity structure of the network. The NWFR model estimates a weighted concurrent regression model for each vertex $\nu_i$ and is given by the following equation:
	
	\begin{equation}\label{eq:mod}
		Y_{\nu_i}(t) = \beta_{\nu_i,0}(t) + \sum_{p=1}^P \int_{\mathcal{S}_p} X_{\nu_i,p}(s) \beta_{\nu_i,p}(t, s) \, ds + \epsilon_{\nu_i}(t), \quad i = 1, \dots, N \; \forall t\in \mathcal{T},
	\end{equation} where $Y_{\nu_i}(t)$  is the response variable at vertex  $\nu_i$ , $\beta_{\nu_i,0}(t)$ is the intercept function at vertex $\nu_i$, $X_{\nu_i,p}(s)$ are the functional covariates observed at vertex $\nu_i$, $\beta_{\nu_i,p}(t, s)$ are the coefficient functions corresponding to each covariate $X_{\nu_i,p}(s)$ and
	$\epsilon_{\nu_i}(t)$ is the error term at vertex $\nu_i$.
	
	For each vertex, the model depends on a diagonal weight matrix $\mathbf{W}_{\nu_i}$, where the element $w_{\nu_i}(j,j)$ ($j=1,\ldots,N$) represents the proximity between vertex $\nu_i$ and $\nu_j$ in the network. The matrix is defined as:
	
	\begin{equation}
		\label{pesi}
		\mathbf{W}_{\nu_i} = [w(i,1), w(i,2), \dots, w(i,N)]^T \mathbf{I}_N;
	\end{equation}
	such that
	\[
	w_{\nu_i}(h,k) =
	\begin{cases}
		0 & \text{if } h \neq k, \\
		w(i,j) & \text{if } h = k = j.
	\end{cases}
	\]
	
	In this paper, the similarity value $w_{\nu_i}(j,j)=w(i,j)$ between vertices $\nu_i$ and $\nu_j$ is computed similarly to \cite{Yamanishi:03}, using the Gaussian descending distance kernel function, a kernel function which is based on the geodesic distance between the vertices, as follows:
	
	\begin{equation}\label{kern}
		w(i,j) = \exp{\left\{-\frac{1}{2}\left(\frac{d_{i,j}}{\theta}\right)^2\right\}},
	\end{equation}
	
	where $d_{i,j}$ is geodesic distance as defined in Eq. (\ref{shortpath}) between vertices $\nu_i$ and $\nu_j$, and $\theta>0$ is a bandwidth parameter selected through cross-validation\cite{Yamanishi:03}. Eq. (\ref{kern}) implies that the more two vertices are far each other the closer to $0$ is the weight. To estimate the functional coefficients $\beta_{\nu_i,p}(t, s)$, we can use basis expansion. Let's represent our curves as linear combinations of basis functions. We approximate the functional data using a set of centered basis functions. 
	
	Let $\mathbf{\Psi}(t) = \left[ \vPs{1}, \dots, \vPs{K}\right]$ be a row vector of basis functions. Similarly, let $\mathbf{\Phi}_{1}(\tx) = \left[\vPh{1,1}, \dots, \vPh{1,K_1}\right]$ be a row vector of basis functions with dimensions $1 \times K_1$, $\mathbf{\Phi}_{2}(\tx) = \left[\vPh{2,1}, \dots, \vPh{2,K_2}\right]$ with dimensions $1 \times K_2$, and so on up to $\mathbf{\Phi}_{\pp}(\tx) = \left[\vPh{\pp,1}, \dots, \vPh{\pp,K_\pp}\right]$ with dimensions $1 \times K_\pp$.
	Let $\mathbf{Y}(t) = \left[\Yv\right]^T$ be the column vector of response variables with dimensions $N \times 1$, and $\mathbf{X}_{p}(s) = \left[\Xv{p}\right]^T$ be the column vector of covariate variables with dimensions $N \times 1$, where $p$ ranges from $1$ to $P$. The curves are expanded as ${X}_{\nu_i,p}(s) = \mathbf{x}_{\nu_i,p} \mathbf{\Phi}_p(\tx)^T$  where $\mathbf{x}_{\nu_i,p}=[x_{\nu_i,p,1},x_{\nu_i,p,2},\ldots,x_{\nu_i,p,K_p}]$ are the components (the scalar coefficients) of the basis in the $\mathbf{\Phi_p(s)}$ vector; ${Y_{\nu_i}(t) = \mathbf{y}_{\nu_i} \mathbf{\Psi}(\ty)^T}$ where $\mathbf{y}_{\nu_i}=[y_{\nu_i,1},y_{\nu_i,2},\ldots,y_{\nu_i,K}]$ are the components (the scalar coefficients) of the basis in the $\mathbf{\Psi}(t)$ vector.
	
	Let's 
	$$\mathbf{X}_p=\left[\begin{matrix}
		x_{\nu_1,p,1} & x_{\nu_1,p,2} & \cdots&x_{\nu_1,p,K_p}\\
		x_{\nu_2,p,1} & x_{\nu_2,p,2} & \cdots&x_{\nu_2,p,K_p}\\
		\cdots & \cdots & \cdots&\cdots\\
		x_{\nu_N,p,1} & x_{\nu_N,p,2} & \cdots&x_{\nu_N,p,K_p}
	\end{matrix}\right]$$
	and 
	$$\mathbf{Y}=\left[\begin{matrix}
		y_{\nu_1,1} & y_{\nu_1,2} & \cdots&y_{\nu_1,K}\\
		y_{\nu_2,1} & y_{\nu_2,2} & \cdots&y_{\nu_2,K}\\
		\cdots & \cdots & \cdots&\cdots\\
		y_{\nu_N,1} & y_{\nu_N,2} & \cdots&y_{\nu_N,K}
	\end{matrix}\right].$$
	
	The estimation procedure of $\beta_{\nu_i,p}(\tx, \ty) = \mathbf{\Phi}_p(\tx) \mathbf{B}_{\nu_i,p} \mathbf{\Psi}(\ty)^T$ for each $p = 1, \dots, \pp$,
	for the vertex $\nu_i$, require to solve the following equation:
	\begin{equation}		
		\label{eq:04}
		(\textbf{X}_p \textbf{J}_{\mathbf{\Phi}_p})^T \mathbf{W}_{\nu_i} \left( \sum_{j=1}^{\pp} \mathbf{X}_j \textbf{J}_{\mathbf{\Phi}_j} \mathbf{B}_{\nu_i,j} \right) \textbf{J}_{\mathbf{\Psi}} = (\mathbf{X}_p \textbf{J}_{\mathbf{\Phi}_p})^T \mathbf{W}_{\nu_i} \mathbf{Y} \,\textbf{J}_{\mathbf{\Psi}}, \quad \forall i = 1, \dots, N,
	\end{equation}
	where $\textbf{J}_{\mathbf{\Phi}_p} = \int {\mathbf{\Phi}}_p^T(t) {\mathbf{\Phi}}_p(t) dt$, a $K_p\times K_p$ matrix, and $\textbf{J}_{\mathbf{\Psi}} = \int {\mathbf{\Psi}}^T(s) {\mathbf{\Psi}}(s) ds$, a $K\times K$ matrix. 
	
	
	Let's define the matrices:
	
	$$\mathbb{X}=\left[\begin{matrix}
		\mathbf{X}_{1}\mathbf{J}_{\mathbf{\Phi}_1}& \mathbf{X}_{2}\mathbf{J}_{\mathbf{\Phi}_2} & \cdots&\mathbf{X}_{P}\mathbf{J}_{\mathbf{\Phi}_P}
	\end{matrix}\right];$$
	of dimensions $N\times\sum_{p=1}^PK_p$, and 
	the matrix:
	$$\mathbb{B}_{\nu_i}=\left[\begin{matrix}
		\mathbf{B}_{\nu_i,1} \\ \mathbf{B}_{\nu_i,2}\\\cdots\\\mathbf{B}_{\nu_i,P}
	\end{matrix}\right];$$
	of dimension $\left(\sum_{p=1}^PK_p\right)\times K$.
	
	\begin{theorem}
		Eq. (\ref{eq:04}) allows as solution:
		\begin{equation}\label{bigsol}
			\mathbb{B}_{\nu_i}=(\mathbb{X}^T\mathbf{W}_{\nu_i}\mathbb{X})^{-1}\mathbb{X}^T\mathbf{W}_{\nu_i}\mathbf{Y}\;;
		\end{equation}
		where $\mathbf{B}_{\nu_i,p}$ is the $p$-th block of $\mathbb{B}$.
	\end{theorem}
	
	\begin{proof}
		One way to “decouple” the Eq. (\ref{eq:04}) is to note that it holds for every group (covariate) $p$. In particular, we have
		
		\[
		(\textbf{X}_p \textbf{J}_{\mathbf{\Phi}_p})^T \mathbf{W}_{\nu_i} \left( \sum_{j=1}^{\pp} \mathbf{X}_j \textbf{J}_{\mathbf{\Phi}_j} \mathbf{B}_{\nu_i,j} \right) \textbf{J}_{\mathbf{\Psi}} = (\mathbf{X}_p \textbf{J}_{\mathbf{\Phi}_p})^T \mathbf{W}_{\nu_i} \mathbf{Y} \,\textbf{J}_{\mathbf{\Psi}}, \quad  \;\forall p=1,\dots,P.
		\]
		
		We assume that $\mathbf{J}_{\mathbf{\Psi}}$ is invertible, then we cancel it on the right so that
		
		\[
		(\textbf{X}_p \textbf{J}_{\mathbf{\Phi}_p})^T \mathbf{W}_{\nu_i} \left( \sum_{j=1}^{\pp} \mathbf{X}_j \textbf{J}_{\mathbf{\Phi}_j} \mathbf{B}_{\nu_i,j} \right) = (\mathbf{X}_p \textbf{J}_{\mathbf{\Phi}_p})^T \mathbf{W}_{\nu_i} \mathbf{Y}\,.
		\]
		
		Now, note that
		
		\[
		\sum_{j=1}^{\pp} \mathbf{X}_j \mathbf{J}_{\mathbf{\Phi}_j}\mathbf{B}_{\nu_i,j}  = \mathbb{X} \mathbb{B}_{\nu_i}.
		\]
		
		Furthermore, since the above equation must hold for every covariate $p$, it is equivalent (when “stacked” over all covariates) to
		
		\[
		\mathbb{X}^T \mathbf{W}_i\,\mathbb{X}\,\mathbb{B}_{\nu_i} = \mathbb{X}^T \mathbf{W}_i\,\mathbf{Y}\,.
		\]
		
		Assuming that $\mathbb{X}^T \mathbf{W}_i \mathbb{X}$ is invertible, we solve for $\mathbb{B}_{\nu_i}$:
		
		\[
		\mathbb{B}_{\nu_i}=(\mathbb{X}^T\mathbf{W}_{\nu_i}\mathbb{X})^{-1}\mathbb{X}^T\mathbf{W}_{\nu_i}\mathbf{Y}\;.
		\]
		
		Once we obtain $\mathbb{B}_{\nu_i}$ (which is of size $\left(\sum_{p=1}^P K_p\right) \times K$), we can extract each block $\mathbf{B}_{\nu_i,p}$.
	\end{proof}

	then $\beta_{\nu_i,p}(t,s)$ in Eq. (\ref{eq:mod}) is given as follows:
	\begin{equation}
		\beta_{\nu_i,p}(t,s)=\mathbf{\Phi}_p(\tx) \mathbf{B}_{p,i} \mathbf{\Psi}(\ty)^T.  
	\end{equation}

	The weight matrix $\textbf{W}_{\nu_i}$ is a diagonal matrix as defined in \ref{kern}.
	It is worth noting that, when the weight matrix $\mathbf{W} = \mathbf{1}_N$ each $\mathbf{W}_{\nu_i}=\mathbf{I}_N$. It implies that all the vertices are equally weighted and the network is fully connected, meaning that the network structure does not influence the model. In this case, all the vertices models are equal each other and to the standard concurrent functional regression model.

	\subsection{Goodness of fit of the model and tests on model coefficients}
	
	To evaluate the goodness of fit of the model, we propose using the classical Root Integrated Mean Squared Error (RIMSE) index, the extensions of the classical $R^2$ measure to functional regression \cite{Yao_mull_05}, and its functional version. The RIMSE measures the mean squared difference between observed and predicted functions integrated across the domain, and is defined as:
	\[
	\text{RIMSE} = \sqrt{\frac{1}{N} \sum_{i=1}^{N} \int_T \left( Y_{\nu_i}(t) - \hat{Y}_{\nu_i}(t) \right)^2 dt}.
	\]
	
	The point-wise $R^2$ index is defined as:
	\[
	R^2(t) = 1 - \frac{\sum_{i=1}^{N} \left( Y_{\nu_i}(t) - \hat{Y}_{\nu_i}(t) \right)^2}{\sum_{i=1}^{N} \left( Y_{\nu_i}(t) - \bar{Y}(t) \right)^2},
	\]
	with $\bar{Y}(t) = \frac{1}{N} \sum_{i=1}^{N} Y_{\nu_i}(t)$.
	
	The integral of the point-wise $R^2$ index, called $\widetilde{R}^2$, is defined as follows:
	\[
	\widetilde{R}^2 = 1 - \frac{\int_T \sum_{i=1}^{N} \left( Y_{\nu_i}(t) - \hat{Y}_{\nu_i}(t) \right)^2 dt}{\int_T \sum_{i=1}^{N} \left( Y_{\nu_i}(t) - \bar{Y}(t) \right)^2 dt}.
	\]

	To test the network structure effect on the variability of the coefficients, we perform a test based on the integrated variance for varying \( \nu_i \).  The main aim is to test the Null Hypothesis that the coefficients \( \hat{\beta}_{\nu_i,k}(t,\,s) \) are stable across the network.
	Formally, this can be expressed as:
	\[
	H_0: \hat{\beta}_{\nu_i,k}(t,\,s) = \hat{\beta}_{\nu_j,k}(t,\,s) \quad 
	\]

	Following \cite{DeBellefonFloch:18}, \cite{Brunsdon:98} and \cite{Yamanishi:03}, We assess the variability of the coefficients \( \hat{\beta}_{\nu_i,k}(t,\,s) \) across different nodes \( \nu_i \) for a fixed \( k \) using the statistic:
	
	\begin{equation}
		v_k = \frac{1}{N} \sum_{i=1}^N \int\int\left(\hat{\beta}_{\nu_i,k}(t, s) - \bar{\beta}_k(t, s) \right)^2 \, dt \, ds,
	\end{equation}

	where \( \bar{\beta}_k(t,\,s) \) represents the mean coefficient function across all nodes, defined as:
	\begin{equation}
		\bar{\beta}_k(t, s) = \frac{1}{N} \sum_{i=1}^{N} \hat{\beta}_{\nu_i,k}(t, s).
	\end{equation}

	A high value of \( v_k \) suggests a significant variation in the coefficients \( \hat{\beta}_k(t,\,s, \nu_i) \) across the nodes, potentially indicating non-stationarity. If the network structure does not influence the coefficient distribution, permuting the node labels should not affect the variance.
	
	To formally test this hypothesis, we conduct Monte Carlo permutation test. We randomly permute the node labels $1,000$ times and compute $1,000$ realizations of the network variance statistic. By comparing the observed statistic \( v_k \) with the null distribution obtained from these permutations, we can assess, through the permutation p-value, whether the network structure significantly impact each $\beta_k$ variability.

	\section{Conformal Inference background}
	\label{sec:04}
	
	Conformal prediction (CP) is a versatile, model-independent framework designed to generate prediction intervals. The approach involves training a model on part of the data, using a held-out calibration set to evaluate prediction errors, and then constructing prediction intervals for new data by applying a quantile of those errors. These intervals are guaranteed to provide coverage, assuming that the data are exchangeable \cite{vovk2005algorithmic}.
	
	In FDA setting \cite{ajroldi2023conformal} and\cite{Diquigiovanni:22} extend SCP to functional responses using a supremum-based nonconformity score computed over a held-out calibration set, where the quantile of these scores determines a uniform threshold. The resulting prediction band for a new input is then constructed by expanding the predicted function uniformly across the domain.
	
	However, this approach depends  on the exchangeability assumption, an assumption often violated in many real‑world domains, particularly those related to spatial processes or network structures with complex dependencies.
	
	To address these limitations, recent research has extended CP to accommodate structured dependencies, including spatial \cite{Mao02042024}, and network-based settings \cite{Cher:2021,lunde2023_2,lunde2023}.
	
	In spatial context, \cite{Mao02042024} address the challenge of predicting at unobserved spatial locations without relying on potentially misspecified model assumptions (like stationarity). They observe that, under an infill asymptotic regime, spatial observations become locally approximately exchangeable for a broad class of processes. Building on this, they develop a local spatial conformal prediction algorithm that constructs valid, nonparametric prediction intervals purely from the data’s empirical variability, without assuming any particular spatial model.
	
	In network settings, observations are interdependent through the graph’s topology, so standard independence assumptions break down and models must explicitly incorporate these connections. 
	Recent advances have shown that conformal prediction can be adapted to maintain valid coverage even when exchangeability is violated in such graph‐structured data \cite{Barber23}.
	Building on this, Lunde et al. \cite{lunde2023} show that further improvements in interval validity and efficiency can be achieved by using network‐aware sampling and stratification techniques that explicitly leverage the underlying graph structure.

	By leveraging the network‐aware sampling and stratification techniques of \cite{lunde2023}, we extend the SPC framework to functional data indexed by graph vertices. Specifically, our method constructs distribution‑free, finite‑sample valid prediction bands for the functional responses at each node by embedding the graph’s dependency structure into the SCP procedure, thereby preserving its formal coverage guarantees.
	
	
	\subsection{Split Conformal Prediction for NWFR}
	
	Predicting an entire functional response at an unobserved node in a network introduces two central challenges: first, uncertainty must be quantified uniformly over the domain of the function; second, the usual assumption of exchangeability is typically violated due to dependencies induced by the graph topology.
	
	To address these issues, given a miscoverage level \(\alpha \in (0,1)\), we construct a prediction band \( C \subset \mathcal{L}_2(T)^p \times \mathcal{L}_2(T) \) such that, for a new vertex \( \nu_j \in D \), the pair \( (X_{\nu_j}, Y_{\nu_j}) \) falls within the band with probability at least \( 1 - \alpha \). That is,
	\[
	\mathbb{P}\left\{ Y_{\nu_j} \in C (X_{\nu_j}) \right\} \geq 1 - \alpha.
	\]
	This probability is taken with respect to a new vertex \( \nu_j \), sampled under the same data-generating process as the calibration set. While observations in a network are not fully independent, we assume approximate exchangeability or conditional independence, achieved through network-aware sampling or stratification techniques.
	
	Let \(D\) denote the set of vertices in the network, and for each \(\nu_i\in D\) let \((X_{\nu_i},Y_{\nu_i})\) denote the pair of functional covariates \(X_{\nu_i}\in\mathcal{L}_2(T)^p\) and response curve \(Y_{\nu_i}\in\mathcal{L}_2(T)\).  Although these pairs are not strictly independent due to the network topology, we assume there exists a partition \(G_1,\dots,G_K\) of \(D\) (e.g.\ via network‐aware stratification) such that, within each \(G_k\),
	\[
	\bigl\{(X_{\nu_i},Y_{\nu_i}):\nu_i\in G_k\bigr\}
	\]
	is exchangeable.  Formally, for any \(m\ge1\), any distinct \(\nu_{i_1},\dots,\nu_{i_m}\in G_k\), and any permutation \(\pi\) of \(\{1,\dots,m\}\),
	\[
	\bigl((X_{\nu_{i_1}},Y_{\nu_{i_1}}),\dots,(X_{\nu_{i_m}},Y_{\nu_{i_m}})\bigr)
	\;\overset{d}{=}\;
	\bigl((X_{\nu_{i_{\pi(1)}}},Y_{\nu_{i_{\pi(1)}}}),\dots,(X_{\nu_{i_{\pi(m)}}},Y_{\nu_{i_{\pi(m)}}})\bigr).
	\]
	By ensuring that at most one vertex from each \(G_k\) enters the calibration set, we recover the exchangeability needed for valid conformal prediction of the curve \(Y_{\nu_j}(t)\) at a new vertex \(\nu_j\) based on its covariates \(X_{\nu_j}\).

	
	We then construct the conformal prediction band for the curve \(Y_{\nu_j}(t)\) at vertex \(\nu_j\) as:
	\[
	C(X_{\nu_j}) = \left\{ Y_{\nu_j}(t) \in \mathcal{L}_2(T) :  \hat{Y}_{\nu_j}(t) - k^S S(t) \leq Y_{\nu_j}(t) \leq \hat{Y}_{\nu_j}(t) + k^S S(t) \right\},
	\]
	where:
	\begin{itemize}
		\item \(\hat{Y}_{\nu_j}(t)\) is the predicted response curve for the unobserved vertex \(\nu_j\),
		\item \(S(t)\) is a modulation function used to allow the band width to vary over the domain \(T\).
		
		\item \(\mathcal{D}_{h}\) denotes a chosen discrepancy or non-conformity measure between the predicted and observed functional responses,
		\item \(k^S\) is the band radius, determined as the \((1 - \alpha)\)-quantile of the empirical distribution of non-conformity scores \(\{R_i : i = T+1, \dots, N\}\),
		\item \(R_i = \mathcal{D}_{h}(\hat{Y}_{\nu_i}, Y_{\nu_i})\) is the non-conformity score for calibration point \(\nu_i\),
	\end{itemize}

	In conformal prediction, the width of the prediction band should ideally be modulated according to the local variability of the data. This modulation can be achieved using a modulation function $S(t)$. Adapting the parameter $k^S$ accordingly, the prediction band can be dynamically expanded or narrowed based on local variability.
	
	For this reason, the non-conformity measures are defined with respect to a modulation function \( S(t) \), which is computed using a calibration set \( \mathcal{Z}_{\mathrm{cal}} \). This set is a subset of the data that is held out from training and is instead used to quantify prediction uncertainty by evaluating discrepancies between predicted and observed responses.

	Formally, the modulation function can be defined as:
	as:
	\[
	S(t) := \sqrt{\frac{1}{|\mathcal{Z}_{\mathrm{cal}}|} \sum_{\nu_k \in \mathcal{Z}_{\mathrm{cal}}} \left( \hat{Y}_{\nu_k}(t) - Y_{\nu_k}(t) \right)^2}.
	\]

	A key ingredient in conformal prediction is the nonconformity measure \cite{pmlr-v204-kato23a}, which quantifies how unusual or "nonconforming" a new observation is relative to the calibration data. In our approach, the nonconformity score incorporates the modulation function \( S(t) \) , allowing the prediction bands to adapt to local heterogeneity in prediction error across the domain.
	
	Considering a general non-conformity score \(\mathcal{D}_{\#}\) defined as:
	\[
	\mathcal{D}_{h}(\hat{Y}_{\nu_j}, Y_{\nu_j}) := \left( \int_T \left( \frac{\hat{Y}_{\nu_j}(t) - Y_{\nu_j}(t)}{S(t)} \right)^h dt \right)^{1/h}.
	\]
	where \( h \geq 1 \) is a parameter controlling the sensitivity of the measure to large deviations. 
	We propose to use \( h = 2 \), that corresponds to the \( L^2 \)-norm which averages the squared errors across the domain.
	
	In the functional literature \cite{Diquigiovanni:22} propose to use a non-conformity measure defined by \( h \to \infty \). It is a supremum-based measure  focusing on the largest error at any point in the domain \( T \), capturing the worst-case discrepancy between the predicted and observed values.
	
	Our choice of using a non-conformity measure that averages the squared errors across the domain aligns with the general principles outlined in \cite{lei2018distribution}and \cite{Mao02042024} and is supported by both theoretical and practical considerations. The squared-error structure corresponds to classical regression diagnostics (e.g., RMSE), enhancing interpretability and comparability with standard approaches.
	The \( L^2 \)-norm offers a balanced trade-off between robustness to local fluctuations and sensitivity to broader structural deviations. It is a smooth and differentiable functional, which facilitates both optimization and theoretical analysis. Moreover, the integral formulation aggregates discrepancies across the entire domain, ensuring that the score reflects global misfit rather than isolated anomalies. In addition compared to high-order norms \( (h \gg 2) \), the use of \( h = 2 \) mitigates numerical instability and variance inflation in finite-sample settings, particularly on discretized or irregular domains.  
	The complete procedure can be described as follows:
	\begin{enumerate}
		\item \textbf{Data Partitioning.}  
		Split the graph‐indexed dataset into a \emph{training set} (for model fitting) and a \emph{calibration set} (for non‐conformity and modulation).  To preserve network topology, first detect communities -groups of nodes within the graph that are more densely connected- using, for example, the Louvain algorithm \cite{Blondel_2008}, and then sample nodes within each community.
		
		\item \textbf{Model Training.}  
		Fit the network‐weighted functional regression (NWFR) model on the training set, using the kernel in Eq.~(\ref{kern}) with network distances (e.g.\ shortest‐path or geodesic) to define neighborhoods.  Denote the fitted predictor by \(\hat Y_{\nu}(t)\).
		
		\item \textbf{Modulation Function.}  
		On the calibration set \(\mathcal{Z}_{\mathrm{cal}}\), compute the local variability
		\[
		S(t) \;=\; \sqrt{\frac{1}{|\mathcal{Z}_{\mathrm{cal}}|} \sum_{\nu_k\in\mathcal{Z}_{\mathrm{cal}}}
			\bigl(\hat Y_{\nu_k}(t)-Y_{\nu_k}(t)\bigr)^2}\,,
		\]
		which will modulate the band width pointwise over \(T\).
		
		\item \textbf{Non-Conformity Measure.}  
		For each calibration vertex \(\nu_i\), define the non-Conformity measure
		\[\mathcal{D}_2\bigl(\hat Y_{\nu_i},Y_{\nu_i}\bigr)
		:=\sqrt{\int_T
			\Bigl(\tfrac{\hat Y_{\nu_i}(t)-Y_{\nu_i}(t)}{S(t)}\Bigr)^2\,dt}\,.
		\]
		
		\item \textbf{Prediction Band Construction.}  
		Let \(k^S\) be the \((1-\alpha)\)-quantile of the calibration scores \(\{R_i\}\).  For a new vertex \(\nu_j\), form the band
		\[
		C_{\mathrm{conf}}(X_{\nu_j})
		= \bigl\{\,y(t):\;\hat Y_{\nu_j}(t)\pm k^S\,S(t)\bigr\},
		\]
		which guarantees \(\mathbb{P}\bigl((X_{\nu_j},Y_{\nu_j}(\cdot))\in C_{\mathrm{conf}}(X_{\nu_j})\bigr)\ge1-\alpha\).
	\end{enumerate}

	
	
	
	
	
	
	\[
	\]

	\subsection{CP performance evaluation}
	The performance evaluation of a conformal prediction method typically involves assessing the \textit{validity} and \textit{efficiency} of the prediction intervals generated by the method. These two aspects are generally opposed meaning that high coverage can be obtained when intervals are wide, but this means that a gain in validity could be generally related to a loss of efficiency. This problem in CP is well-known as the \textit{validity-efficiency} trade-off \cite{pmlr-v204-kato23a}. For this reason, we will use \cite{Diana:23} the following indices:  
	\begin{itemize}
		\item 
		a global empirical coverage index $Cov_G$, a measure of \textit{Global} validity, counting the proportion of observed curves completely contained in the prediction band. Formally, being $PB_i= [I_{l_i}, I_{u_i}]$ the prediction band for the $i-th$ observation, $Cov_G$ is defined as follows
		\begin{equation}\label{eq:GLO_cov}					
			Cov_G = \frac{1}{N}\sum_i\mathbb{I}\left( Y_{\nu_i}\in PB_i \right),					
		\end{equation}			
		where
		\[
		\mathbb{I}\left( Y_{\nu_i} \in PB_i \right) = 
		\begin{cases} 
			1 & \text{if } I_{l_i}(t) \leq Y_{\nu_i}(t) \leq I_{u_i}(t)\;\; \forall t\in[0,1]\\
			0 & \text{otherwise}
		\end{cases}\,\,;
		\]
		$Cov_G$ assumes values in $[0,1]$, and the more it is close to $1$ the more the CP covers all the observed dependent data;
		\item 
		since $Cov_G$ has the strong requirement coverage for each curve, namely the curve must completely covered by the prediction band, in order to soften such requirement, we propose a new measure of \textit{Local} coverage, $Cov_L$, which quantifies the proportion of the function domain satisfying $ Y_{\nu_i}(t)\in PB_i(t)$. $Cov_L$ is defined as follows:

		\begin{equation}\label{eq:loc_cov}					
			Cov_L = \frac{1}{N}\sum_i\int_T\mathbb{I}\left( Y_{\nu_i}(t)\in PB_i(t) \right)dt					
		\end{equation}			
		where 
		
		\[
		\mathbb{I}\left( Y_{\nu_i}(t) \in PB_i(t) \right) = 
		\begin{cases} 
			1 & \text{if } I_{l_i}(t) \leq Y_{\nu_i}(t) \leq I_{u_i}(t) \\
			0 & \text{otherwise}
		\end{cases}\,\,;
		\]
		similar to $Cov_G$, $Cov_L$ is defined in $[0,1]$ and usually $Cov_L\geq Cov_G$;
		\item 
		the \textit{Average Bandwidth}, defined as follows: 
		\begin{equation}\label{eq:ABW}
			ABW=\frac{1}{N}\sum_i\int_T\left(I_{u_i}(t) - I_{l_i}(t)\right)dt\, ;
		\end{equation}
		returns the average of average bandwidths for each observation and can be considered a measure of efficiency of the CP method;
		\item 	
		the functional version of the Interval Score $S^{int}_{\alpha}$  \cite{Gneiting:07}, defined as:
		\begin{equation}\label{eq:11}
			S^{int}_{\alpha}=\frac{1}{N}\sum_{i} \int_TS^{int}_{\alpha\,i} (t)dt,
		\end{equation}
		where
		\begin{equation}
			S^{int}_{\alpha\,i}(t)=\left[I_{u_i}(t) - I_{l_i}(t)\right] + \frac{2}{\alpha}\left[I_{l_i}(t)-Y_{\nu_{i}}(t)\right]_+ + \frac{2}{\alpha}\left[Y_{\nu_{i}}(t) - I_{u_i}(t)\right]_+ 
		\end{equation}
		where $[I_l(t), I_u(t)]$ is the prediction band, $Y^s(t)$ contains the observations $Y_{\nu_1}(t),\dots,Y_{\nu_{N}}(t)$, and $z_+=z\lor 0$ denotes the “positive part”. The minimum of $S^{int}_{\alpha}$ is equal to the $ABW$ index when the data are completely contained within the band; otherwise, $S^{int}_{\alpha}$ measures the distance of the data from the band. A smaller $S^{int}_{\alpha}$ is desirable as this rewards both high coverage and narrow intervals. In the remainder of the paper we consider a value of $\alpha=0.05$.
		
	\end{itemize}

	\section{Simulation study}
	\label{sec:05}
	
	The synthetic network structure consists of four interconnected communities of nodes, forming a single connected graph, which are generated using a stochastic block model \cite{sbm} having weighted links \cite{Wsbm}. Each node is related to a pair of functional data: one representing a response variable and the other a covariate, which together capture the functional attributes of the nodes. To ensure robust evaluation and reliable average results, we generated $100$ random replicates of this structure. To investigate a range of configurations, we varied key characteristics, creating 12 distinct scenarios. Each community was modeled as a random undirected graph, with intra link probabilities chosen randomly from the interval $[0.6, 0.8]$. These scenarios were defined by three primary parameters: \textit{Edge Weights} (EW), \textit{Order of the Communities} (OC), and \textit{Connectivity Between Communities} (CBC). For a visual representation of the generated structure, please refer to the Supplementary Materials.
	
	Each parameter is set as follows:
	\begin{enumerate}
		\item \textbf{Edge Weights} ({\bf LW}). All edge weights are equal to 1 (Case \textit{One}); the edge weights are randomly generated within the range $[0.1, 0.9]$ (Case \textit{Random}); edge weights within the communities are generated in the range $[0.3, 0.6]$, while weights between communities are in the range $[0.6, 0.9]$ (Case \textit{InOut}).        
		\item \textbf{Order of the Communities} ({\bf OC}). 
		Each Community has 25 nodes (Case \textit{Equal}); the Communities have a different number of nodes, but the total number of nodes of the graph is $100$ (Case \textit{Different}).
		\item \textbf{Connectivity Between Communities} ({\bf CBC}).
		Communities inter link probability is $0.2$ ( Case \textit{High}), respectively the $0.5$ (Case \textit{High}). 
	\end{enumerate}
	For each simulated scenario, functional attributes are generated using the model:
	\begin{equation}
		\label{eq:conc}
		Y_{\nu_i}(t)=\int_T{X_{\nu_i,k}(s)\beta_{\nu_i,1}(t,\,s)ds}+\epsilon_{\nu_i}(t), \mbox{ } i=1,\ldots,100,
	\end{equation}
	where the basis functions, functional parameters, functional covariates, response values, and errors are defined as follows:
	\begin{itemize}
		\item 
		{\em Basis functions}: For both models,  {$X_{\nu_i}(s)$}, $\beta(t,s, \nu_i)$ , $Y_{\nu_i}(t)$ and $\epsilon_{\nu_i}(t)$  are expanded in terms of the same B-spline basis functions ${\phi} (s)=(\phi_1(s),\ldots, \phi_k(s))^T$, with $K=21$, and  {$s$} and $t \in [0,1]$. 
		\item 
		{\em Functional covariates}: The covariate is defined as {$X_{\nu_i}(s)=\sum_{j=1}^k X_{\nu_ij}\phi_j(s)$}, with $\phi_j(s)$ defined as above. The coefficients {$X_{\nu_ij}$} with $i=1,\ldots,N$ and $j=1,\ldots,K$ are generated by $\textbf{X} \sim {NMV} \left({\mu}_{k \times 1}, \mathbf{I}_{k \times k}\right)$, where ${\mu}_{K \times 1}=(0, \dots, 0)^T$;
		\item 
		{\em Functional parameters}: We set $\beta_{\cdot,1}(t,s)=\sum_{j=1}^K\sum_{h=1}^K\phi_h(s)b_{j,h,1}\phi_j(t)= {\phi}^T (t) \textbf{B}_1 {\phi}^T (s),$ where $\textbf{B}_1 \sim {NMV} \left({\mu}_{K\times 1}, \mathbf{I}_{K \times K}\right)$ is a matrix of B-spline basis generated with ${\mu}_{K \times 1}=(0, \dots, 0)^T$ and we define: $\beta_{\nu_i,1}(t,s)=\beta_{\cdot,1}(t,s)+2, \forall i=1,\dots,25$; $\beta_{\nu_i,1}(t,s)=-\beta_{\cdot,1}(t,s)-2, \forall i=26,\dots,50$; $\beta_{\nu_i,1}(t,s)=2\beta_{\cdot,1}(t,s)-1, \forall i=51,\dots,75$; $\beta_{\nu_i,1}(t,s)=-2\beta_{\cdot,1}(t,s)+1, \forall i=76,\dots,100$.
		\item  
		{\em Response variables}: We use the model (\ref{eq:conc})  to generate the values for the response $Y_{s_i}(t)$  for $i=1,\ldots,n$.
		\item 
		{\em Errors of the model}: {${\epsilon}_V(t)= (\epsilon_{\nu_1}(t),\ldots,\epsilon_{\nu_N}(t)) \sim \rm{NMV}({\mu}_{K \times 1}, 10^{-4}\mathbf{I}_{K \times k})$}, with ${\mu}_{K \times 1}=(0, \dots, 0)^T$.
	\end{itemize}
	We assess the predictive performance of the NWFR model compared to the FR model by analyzing the average results from 100 simulation runs across the different configurations. Table \ref{tab:sim_classic_GOF} presents the average values of key GOF metrics, namely, $RISME$, $R^2$ and $\widetilde{R}^2$ calculated over 12 scenarios for the classic functional regression runs (namely, no network structure is considered), while in Table\ref{tab:sim_NWFR_GOF} are reported the same indices for the NWFR model runs.
	\begin{table}[!h]
		\centering
		\resizebox{\ifdim\width>\linewidth\linewidth\else\width\fi}{!}{
			\begin{tabular}[t]{lllr>{\small\itshape}cc>{\small\itshape}cc>{\small\itshape}c}
				\toprule
				\multicolumn{3}{c}{\textbf{Scenario}}  &\multicolumn{6}{c}{\textbf{GOF average indices and st. dev. $\sigma$}}\\
				\textbf{EW}&\textbf{OC}&\textbf{CBC}  & \textit{av.} $RIMSE$ & $\sigma_{RIMSE}$ &\textit{ av.} $R^2\%$ & $\sigma_{R^2}$ & \textit{av.} $\widetilde{R}^2\%$ & $\sigma_{\widetilde{R}^2}$   \\
				\midrule
				One     &Equal      &Low& 0.339 & 0.028 & 21.29 & 4.17 & 21.36 & 4.89\\
				&           &High& 0.340 & 0.027 & 20.22 & 4.44 & 20.09 & 5.14\\
				&Different  &Low& 0.322 & 0.026 & 24.04 & 5.00 & 23.56 & 5.79\\
				&           &High& 0.318 & 0.026 & 23.91 & 5.25 & 23.48 & 6.06\\
				Random  &Equal      &Low& 0.340 & 0.026 & 21.00 & 4.70 & 21.03 & 5.39\\
				&&High& 0.339 & 0.031 & 21.16 & 4.93 & 21.16 & 5.72\\
				&Different  &Low& 0.320 & 0.030 & 23.88 & 5.36 & 23.37 & 6.02\\
				&&High& 0.319 & 0.031 & 23.89 & 5.67 & 23.36 & 6.43\\
				In-Out  &Equal      &Low& 0.340 & 0.026 & 21.00 & 4.70 & 21.03 & 5.39\\
				&&High& 0.339 & 0.031 & 21.16 & 4.93 & 21.16 & 5.72\\
				&Different  &Low& 0.320 & 0.030 & 23.88 & 5.36 & 23.37 & 6.02\\
				&&High& 0.321 & 0.028 & 24.23 & 5.67 & 23.84 & 6.47\\
				\bottomrule
		\end{tabular}}
		\caption{\label{tab:sim_classic_GOF}Functional regression GOF indices}
	\end{table}

	\begin{table}[!h]
		\centering
		\resizebox{\ifdim\width>\linewidth\linewidth\else\width\fi}{!}{
			\begin{tabular}[t]{lllr>{\small\itshape}cc>{\small\itshape}cc>{\small\itshape}c}
				\toprule
				\multicolumn{3}{c}{\textbf{Scenario}}  &\multicolumn{6}{c}{\textbf{GOF average indices and st. dev. $\sigma$}}\\
				\textbf{EW}&\textbf{OC}&\textbf{CBC}  & \textit{av.} $RIMSE$ & $\sigma_{RIMSE}$ &\textit{ av.} $R^2\%$ & $\sigma_{R^2}$ & \textit{av.} $\widetilde{R}^2\%$ & $\sigma_{\widetilde{R}^2}$   \\
				\midrule
				One     &Equal      &Low	& 0.036 & 0.038 & 98.17 & 3.21 & 98.14 & 3.32\\
				&           &High		& 0.044 & 0.026 & 98.21 & 2.73 & 98.20 & 2.78\\
				&Different  &Low	& 0.031 & 0.037 & 98.31 & 2.52 & 98.28 & 2.60\\
				&           &High		 & 0.046 & 0.028 & 97.80 & 3.04 & 97.74 & 3.23\\
				Random  &Equal  &Low	& 0.096 & 0.030 & 93.26 & 4.53 & 93.20 & 4.79\\
				&&High					 & 0.136 & 0.024 & 86.99 & 4.41 & 86.96 & 4.60\\
				&Different  &Low	 & 0.089 & 0.034 & 93.18 & 6.93 & 93.13 & 7.15\\
				&&High					 & 0.129 & 0.025 & 86.94 & 5.08 & 86.92 & 5.37\\
				In-Out  &Equal&Low		& 0.241 & 0.027 & 60.05 & 7.31 & 59.90 & 7.71\\
				&&High					& 0.058 & 0.030 & 97.14 & 3.92 & 97.12 & 3.97\\
				&Different  &Low	& 0.063 & 0.032 & 96.31 & 3.10 & 96.24 & 3.20\\
				&&High					& 0.071 & 0.033 & 95.45 & 4.12 & 95.36 & 4.26\\
				\bottomrule
		\end{tabular}}
		\caption{\label{tab:sim_NWFR_GOF}NWFR Model GOF indices.}
	\end{table}
	
	For the functional regression model (Tab. \ref{tab:sim_classic_GOF}), the average $R^2$ and $\widetilde{R}^2$ values are relatively low, ranging approximately between 20–24\%, with relatively stable standard deviations (around 4–6\%). The average RIMSE values are around 0.32–0.34, indicating moderate prediction errors across scenarios. These results suggest limited explained variability and moderate predictive accuracy in the classic framework.
	
	In contrast, the NWFR model (Table \ref{tab:sim_NWFR_GOF}) shows a clear improvement in fit quality, with average $R^2$ and $\widetilde{R}^2$ values exceeding 86\% in nearly all scenarios and reaching above 98\% in the best cases. The $RIMSE$ values are substantially lower, typically below 0.1 in most scenarios, demonstrating much lower prediction errors. While some increase in variability is observed under the “Random” and “In-Out” conditions, overall the NWFR model consistently outperforms the classic model, particularly in terms of explained variability and error reduction.
	
	These results highlight the superior accuracy and efficiency of the NWFR approach across diverse network weights configurations and confirm its advantage over the classic functional regression model.

	We assess the performance of conformal prediction (CP) bands applied to the classical functional regression (Tab. \ref{tab:SIM_cp_classic}) and network-weighted functional regression (NWFR) (Tab. \ref{tab:sim_cp_NWFR}) models under the 12 scenarios. The evaluation considers global and local coverage probabilities, average band width (ABW), and integrated scores across different scenarios and the two conformity measures, $\mathcal{D}_2$ and $\mathcal{D}_\infty$. The results are averaged over 100 simulation repetitions to ensure robust estimates. The following tables summarize the average performance indices and their interpretation, we considered just one table for the classical functional regression because, in that case, no network structure is assumed and the results are for all the 12 scenario as expected.
	
	\begin{table}[!h]
		\centering
		\begin{tabular}[t]{ccccc}
			\toprule
			&\multicolumn{4}{c}{Funct. Regression}\\
			$\mathcal{D}_h$ &  $Cov_G\%$ & $Cov_G\%$ & $ABW$ & $S^{int}_\alpha$\\
			\midrule
			$\mathcal{D}_2$ & 77.0 & 92.6 & 1.785 & 2.736 \\
			$\mathcal{D}_\infty$ & 90.7 & 97.5 & 2.313 & 2.653\\
			\bottomrule
		\end{tabular}
		\caption{\label{tab:SIM_cp_classic}Performance indices for Conformal Prediction for Functional Regression. Averages of 100 repetitions on the 12 scenarios.}
	\end{table}

	\begin{table}[!h]
		\centering
		\begin{tabular}[t]{cccccccc}
			\toprule
			\multicolumn{3}{c}{\bf Scenario}&&\multicolumn{4}{c}{NWFR}\\
			EW & OC & CBC & $\mathcal{D}_h$  & $Cov_G\%$ & $Cov_L\%$ & $ABW$ & $S^{int}_\alpha$\\
			\midrule
			&  &  & $\mathcal{D}_2$ & 85.8 & 95.8 & 3.821 & 5.009\\
			
			&  & \multirow{-2}{*}{\centering\arraybackslash Low} & $\mathcal{D}_\infty$ & 96.0 & 98.9 & 5.051 & 5.435\\
			
			&  &  & $\mathcal{D}_2$  & 85.6 & 95.4 & 4.760 & 6.303\\
			
			& \multirow{-4}{*}[0.5\dimexpr\aboverulesep+\belowrulesep+\cmidrulewidth]{\centering\arraybackslash Equal} & \multirow{-2}{*}{\centering\arraybackslash High} & $\mathcal{D}_\infty$ & 95.7 & 98.8 & 6.290 & 6.791\\
			
			&  &  & $\mathcal{D}_2$ & 84.4 & 95.4 & 3.034 & 4.163\\
			
			&  & \multirow{-2}{*}{\centering\arraybackslash Low} & $\mathcal{D}_\infty$  & 95.7 & 98.7 & 4.064 & 4.414\\
			
			&  &  & $\mathcal{D}_2$ & 84.2 & 95.4 & 4.666 & 6.253\\
			
			\multirow{-8}{*}[1.5\dimexpr\aboverulesep+\belowrulesep+\cmidrulewidth]{\centering\arraybackslash One} & \multirow{-4}{*}[0.5\dimexpr\aboverulesep+\belowrulesep+\cmidrulewidth]{\centering\arraybackslash Different} & \multirow{-2}{*}{\centering\arraybackslash High} & $\mathcal{D}_\infty$  & 96.0 & 99.0 & 6.302 & 6.767\\
			\cmidrule{1-8}
			&  &  & $\mathcal{D}_2$  & 82.9 & 95.1 & 5.538 & 7.444\\
			
			&  & \multirow{-2}{*}{\centering\arraybackslash Low} & $\mathcal{D}_\infty$  & 95.7 & 98.8 & 7.255 & 7.838\\
			
			&  &  & $\mathcal{D}_2$ & 83.3 & 95.0 & 6.654 & 8.975\\
			
			& \multirow{-4}{*}[0.5\dimexpr\aboverulesep+\belowrulesep+\cmidrulewidth]{\centering\arraybackslash Equal} & \multirow{-2}{*}{\centering\arraybackslash High} & $\mathcal{D}_\infty$ & 95.6 & 98.8 & 8.539 & 9.362\\
			
			&  &  & $\mathcal{D}_2$  & 80.8 & 95.1 & 5.155 & 7.301\\
			
			&  & \multirow{-2}{*}{\centering\arraybackslash Low} & $\mathcal{D}_\infty$  & 95.7 & 99.1 & 6.972 & 7.674\\
			
			&  &  & $\mathcal{D}_2$  & 80.5 & 94.9 & 6.160 & 8.890\\
			
			\multirow{-8}{*}[1.5\dimexpr\aboverulesep+\belowrulesep+\cmidrulewidth]{\centering\arraybackslash Random} & \multirow{-4}{*}[0.5\dimexpr\aboverulesep+\belowrulesep+\cmidrulewidth]{\centering\arraybackslash Different} & \multirow{-2}{*}{\centering\arraybackslash High} & $\mathcal{D}_\infty$  & 95.7 & 99.0 & 8.198 & 9.233\\
			\cmidrule{1-8}
			&  &  & $\mathcal{D}_2$  & 81.5 & 94.8 & 7.496 & 10.716\\
			
			&  & \multirow{-2}{*}{\centering\arraybackslash Low} & $\mathcal{D}_\infty$ & 95.6 & 98.9 & 9.641 & 11.040\\
			
			&  &  & $\mathcal{D}_2$ & 84.4 & 95.2 & 5.386 & 7.121\\
			
			& \multirow{-4}{*}[0.5\dimexpr\aboverulesep+\belowrulesep+\cmidrulewidth]{\centering\arraybackslash Equal} & \multirow{-2}{*}{\centering\arraybackslash High} & $\mathcal{D}_\infty$ & 95.6 & 98.7 & 6.991 & 7.571\\
			
			&  &  & $\mathcal{D}_2$  & 81.7 & 95.0 & 4.914 & 6.839\\
			
			&  & \multirow{-2}{*}{\centering\arraybackslash Low} & $\mathcal{D}_\infty$  & 95.5 & 99.0 & 6.605 & 7.197\\
			
			&  &  & $\mathcal{D}_2$ & 81.6 & 95.2 & 5.333 & 7.380\\
			
			\multirow{-8}{*}[1.5\dimexpr\aboverulesep+\belowrulesep+\cmidrulewidth]{\centering\arraybackslash In Out} & \multirow{-4}{*}[0.5\dimexpr\aboverulesep+\belowrulesep+\cmidrulewidth]{\centering\arraybackslash Different} & \multirow{-2}{*}{\centering\arraybackslash High} & $\mathcal{D}_\infty$  & 95.7 & 99.0 & 7.141 & 7.822\\
			\bottomrule
		\end{tabular}
		\caption{\label{tab:sim_cp_NWFR}Performance indices for Conformal Prediction for NWFR (Network Weighted Functional Regression). Averages on 100 repetitions.}
	\end{table}

	For the functional regression model (Tab. \ref{tab:SIM_cp_classic}), coverage at the global level ($Cov_G\%$) and local level ($Cov_L\%$) increases when using $\mathcal{D}_\infty$ compared to $\mathcal{D}_2$, at the cost of slightly larger average band width (ABW). The average integrated score $S^{int}_\alpha$ remains comparable between metrics, reflecting a balance between coverage and efficiency.
	
	For the NWFR model (Tab. \ref{tab:sim_cp_NWFR}), a similar pattern is observed: $\mathcal{D}_\infty$ consistently yields higher global and local coverage across all scenarios, while $\mathcal{D}_2$ offers narrower prediction bands (lower ABW), reflecting higher efficiency but slightly lower coverage. The results also show that NWFR achieves substantially better coverage (typically exceeding 95\%) compared to the functional regression model, especially when $\mathcal{D}_\infty$ is used, although at the cost of wider bands.
	
	Overall, the results highlight the well-known coverage-efficiency trade-off in conformal prediction, with $\mathcal{D}_\infty$ prioritizing coverage and $\mathcal{D}_2$ offering better band efficiency. NWFR, in particular, provides substantial improvements in predictive reliability across complex spatial scenarios.

	\section{Application on real data}
	\label{sec:06}	
	The micro-climate of work environment is a topic that falls in the framework of physical risks in the workplace. As numerous studies have shown, adverse micro-climatic conditions, such as humidity, significant temperature fluctuations, or air currents, can exert a negative impact not only on the health of employees but also on their work performance (\cite{Andersson:06}, \cite{ Vimalanathan:14}). 			
	The effects of incorrect micro-climatic conditions on individuals can range from a mere perception of thermal discomfort to more profound impacts on vital functions and work-related activities (\cite{Kenny:08}, \cite{Parsons:03}, \cite{Sawka:84}). Maintaining an adequate micro-climate is considered an essential requirement for ensuring the health and safety of workers. 
	Predicting and assessing alterations in the micro-climate within a working environment has become essential for upholding health and safety standards in the workplace (\cite{Fang:98},  \cite{Hollands:22}). To this end, the Intel indoor dataset\footnote{The dataset is freely available the following URL address:\\ \url{http://db.csail.mit.edu/labdata/labdata.html}} emerges as one of the most relevant and practical sources of real-world data for the study of micro-climates in working environments in the Intel Berkeley Research lab. The dataset contains a total of 2.3 million readings collected by 54 sensors. Each reading includes information such as: (\textit{date}) Day of the measurement in format month-day-year (from Feb. 28th to Apr. 5th); (\textit{timestamp}) time of the recording in format HH:MM:SS; \textit{node id}: Id of the sensor; (\textit{temperature}) temperature recorded in Celsius degrees; (\textit{humidity}) relative Humidity, ranging in 0-100\%; (\textit{light}) the intensity of Lux measured in $\log(\text{Lux})$; (\textit{voltage}) Voltage ranging from 2-3 volts (Figure \ref{fig:var}).
	\begin{figure}[htbp]
		a)\includegraphics[scale=.15]{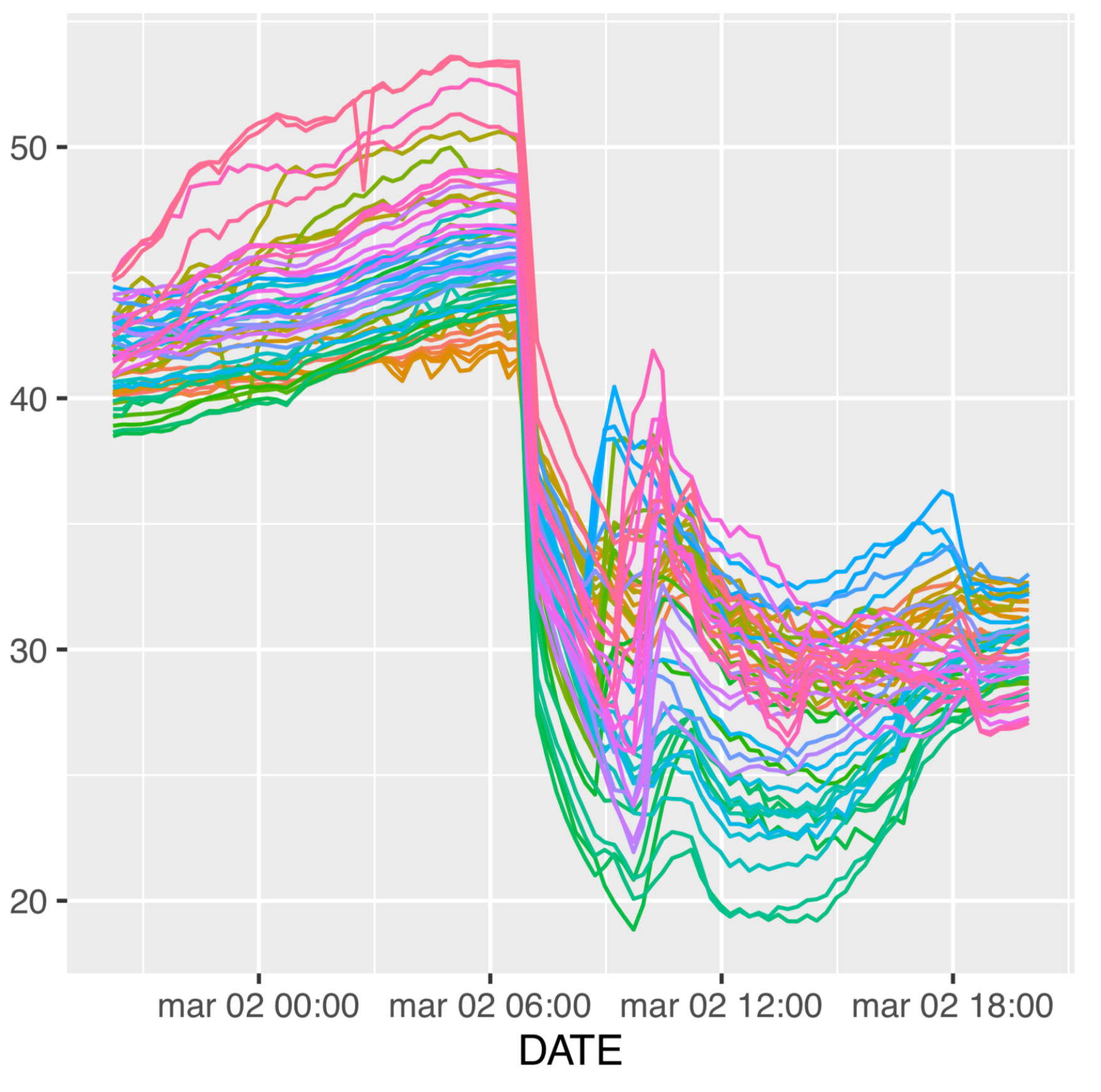}
		b)\includegraphics[scale=.15]{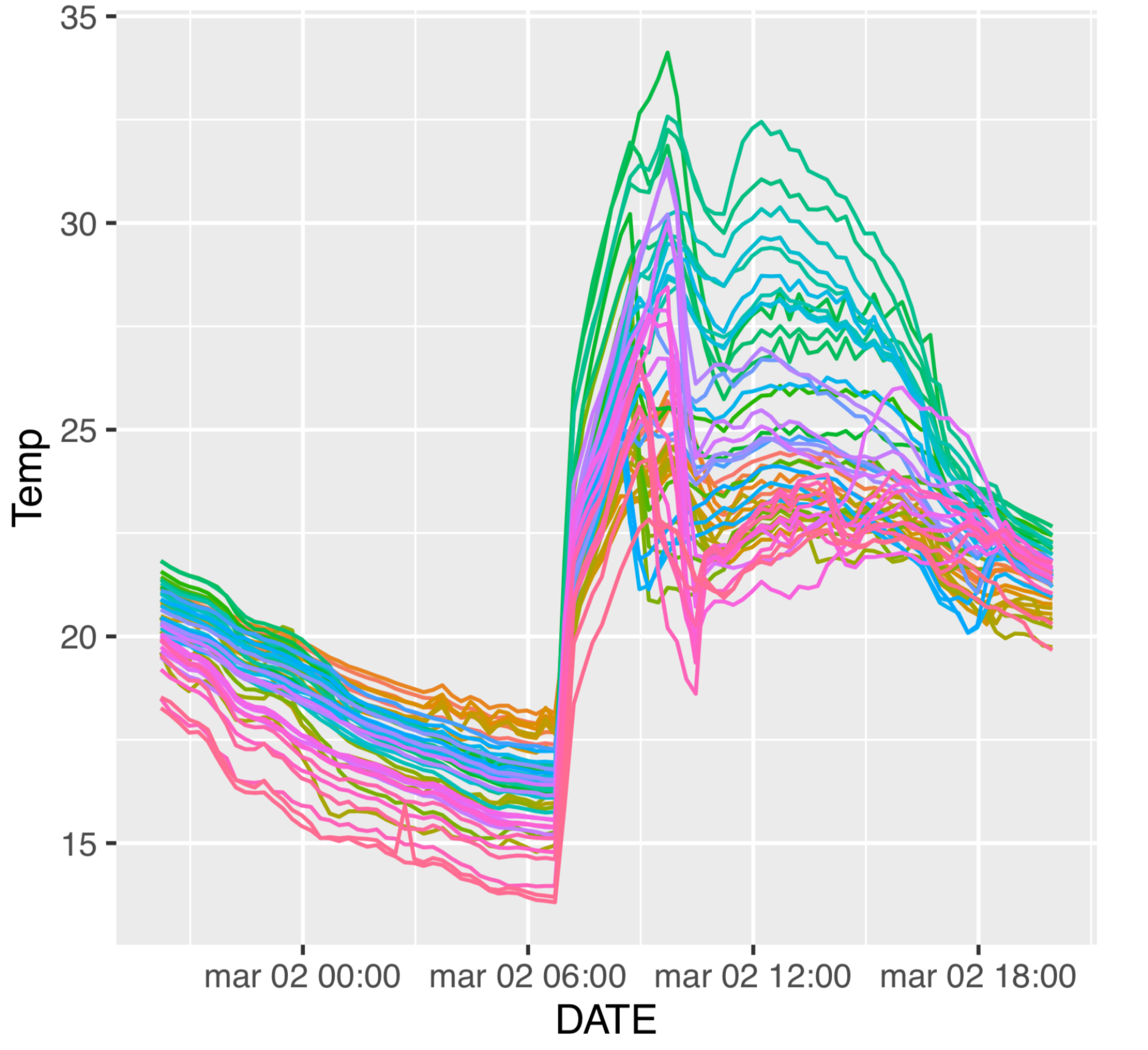}
		c)\includegraphics[scale=.15]{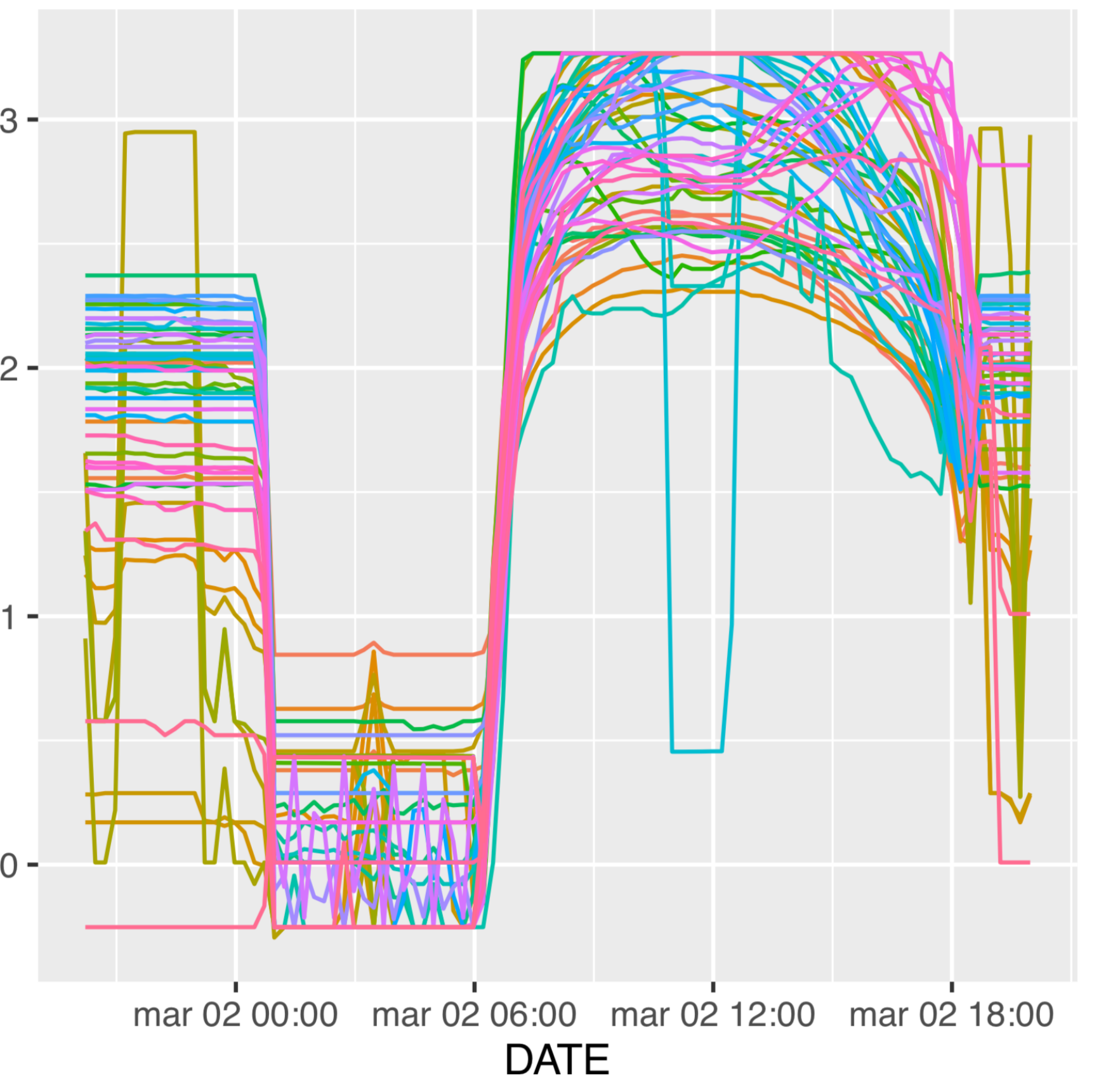}
		\caption{The environmental variables: Relative Humidity (a), Temperature (b) and log(Lux) (c).}
		\label{fig:var}
	\end{figure}
	
	Readings were collected by  using \verb|Mica2Dot| sensors during a period spanning from February 28th to April 5th, 2004 every 31 seconds. Final data was obtained using the \verb|TinyDB| in-network query processing system, which is built on the \verb|TinyOS| platform. Authors provided data without a quality check. For this reason, we consider only data from (Starting Date) to (Ending Date) where neither missing nor anomalous data.
	
	The sensors' data also includes the devices' coordinates, which are expressed in meters and relative to the upper right corner of the lab, and a connectivity probability between each pair of sensors.
	Our main focus in this study is to evaluate the performance of the proposed NWFR model in predicting micro-climate changes in the lab by studying the impact of temperature and light on the relative humidity, in particular we focus on analysing the relation $HUM \sim TEMP + LUX$ from 21:00 of the 1st of Mar. 2004 to 21:00 of the 2nd of Mar 2004.
	Exploratory analysis of the Intel dataset highlighted  there are $2$ sensors that have missing readings. We have removed these sensors and considered the remaining $52$ sensors for further analysis. $K-$ Nearest Neighbors algorithm with $k=3$ has been used as a further instrument to fill in the gaps. 
	After this preprocessing step,  the data has been divided into four different weeks and the analysis has been focused on time windows of 15 minutes length. Data are also provided with information about the average probability that the signal from a sensor is correctly received from another one. We considered this information as a measure of the quality of connectivity between two sensors. This measure has been used for the definition of the weights on the edges between vertices in a network, where vertices represent the sensors. 		 
	Formally let $G=(\mathcal{V},\mathcal{E},F_\mathcal{V},F_\mathcal{E})$ be a the network sensor, where: $\mathcal{V}$ is the set of sensors; $\mathcal{E} \subseteq \mathcal{V}\times \mathcal{V}$  are the links between (all the) sensors\footnote{We assume that the sensors network is shaped as a complete graph.}; $F_\mathcal{V}:\mathcal{V}\rightarrow  \mathcal{L}_2(T)^3$ are the functional attributes  $HUM$, $LUX$ and $TEMP$ observed in each node\footnote{We have assumed bspline basis (the optimal number of basis functions was determined by cross validation \cite{Ramsay:05}).}; $F_\mathcal{E}:\mathcal{E}\rightarrow \mathbb{R}^+$ are weights associated with each edge defined as the probability of one sensor of receiving the signal from the other, formally $f_{i,j}=-log(p_{r_{ij}})$ where $p_{r_{ij}}$ is the probability of sensor $i$ of receiving the signal from sensor $j$ coming from the connectivity probability provided in the dataset. Since in the data comes with the spatial coordinates of the sensors, we compare the proposed NWFR model with the GWFR model, where the weights depend solely on the spatial location.
	The regression coefficient functions $\beta_1(t, s, \nu_{24})$ and $\beta_2(t, s, \nu_{24})$ for a chosen sensor (number 24) obtained by the NWFR and GWFR models show the main characteristics of the phenomena.
	\begin{figure}
		\centering
		a)\includegraphics[scale=0.07]{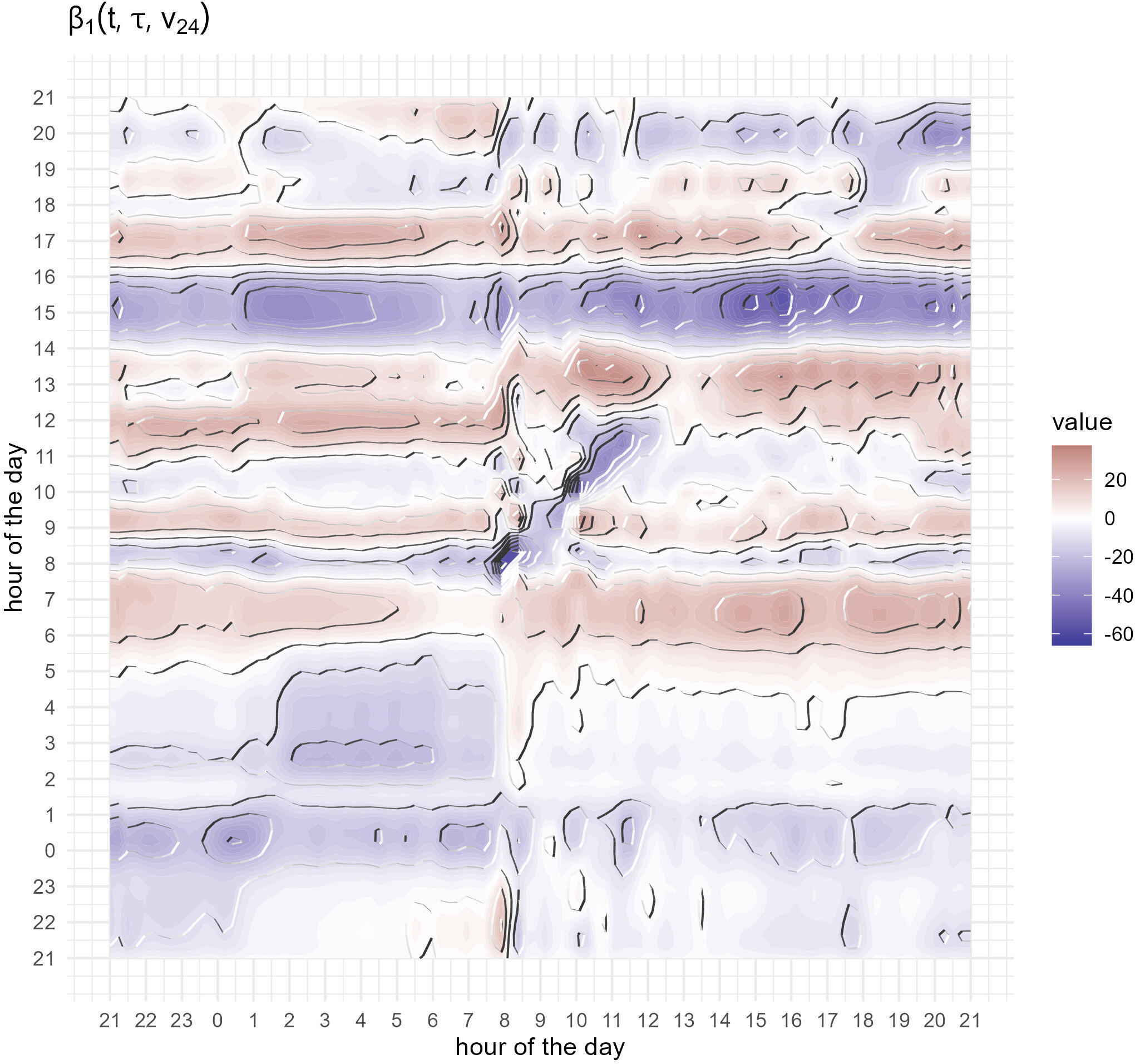}
		b)\includegraphics[scale=0.07]{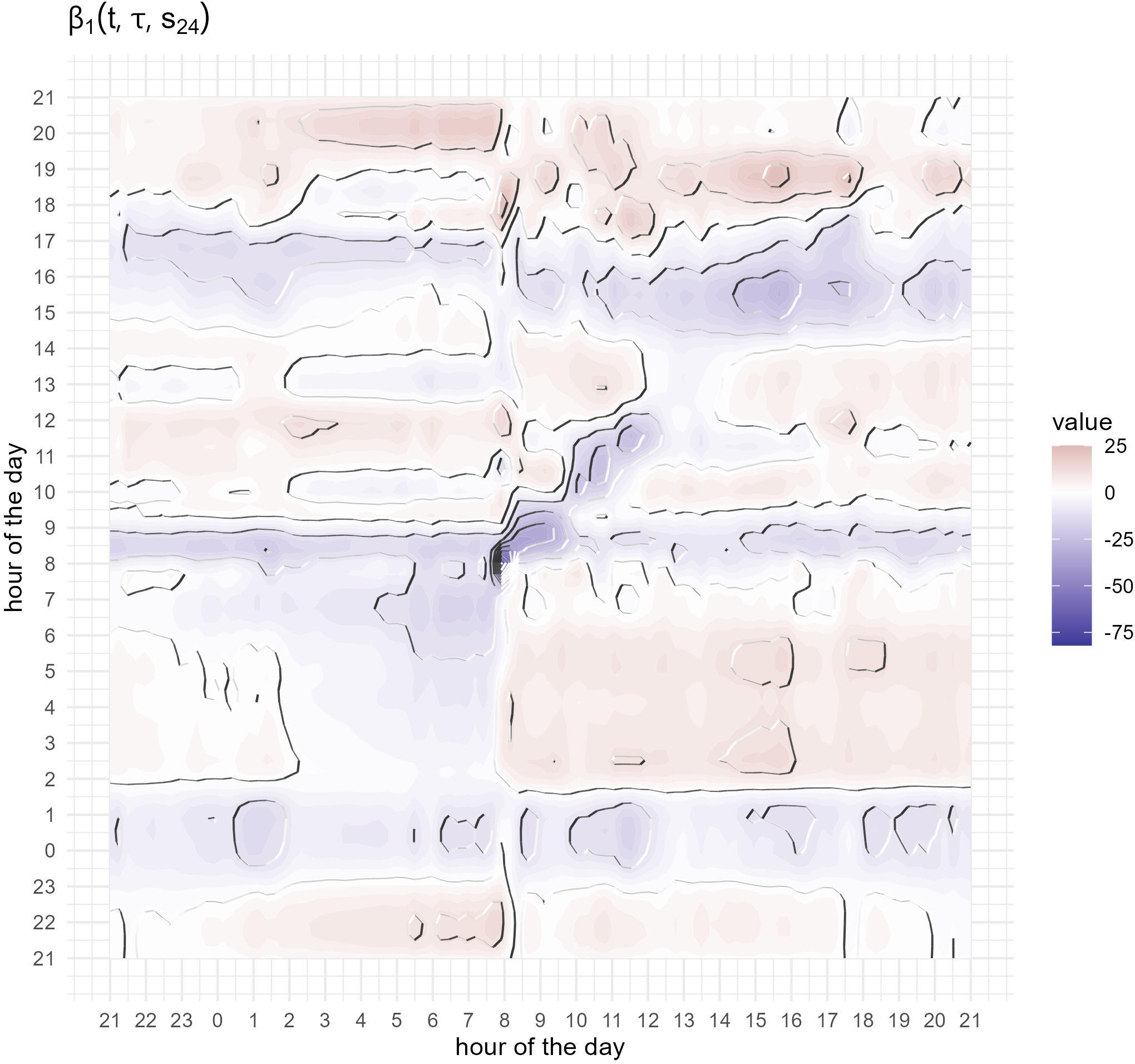}
		\\
		c)\includegraphics[scale=0.07]{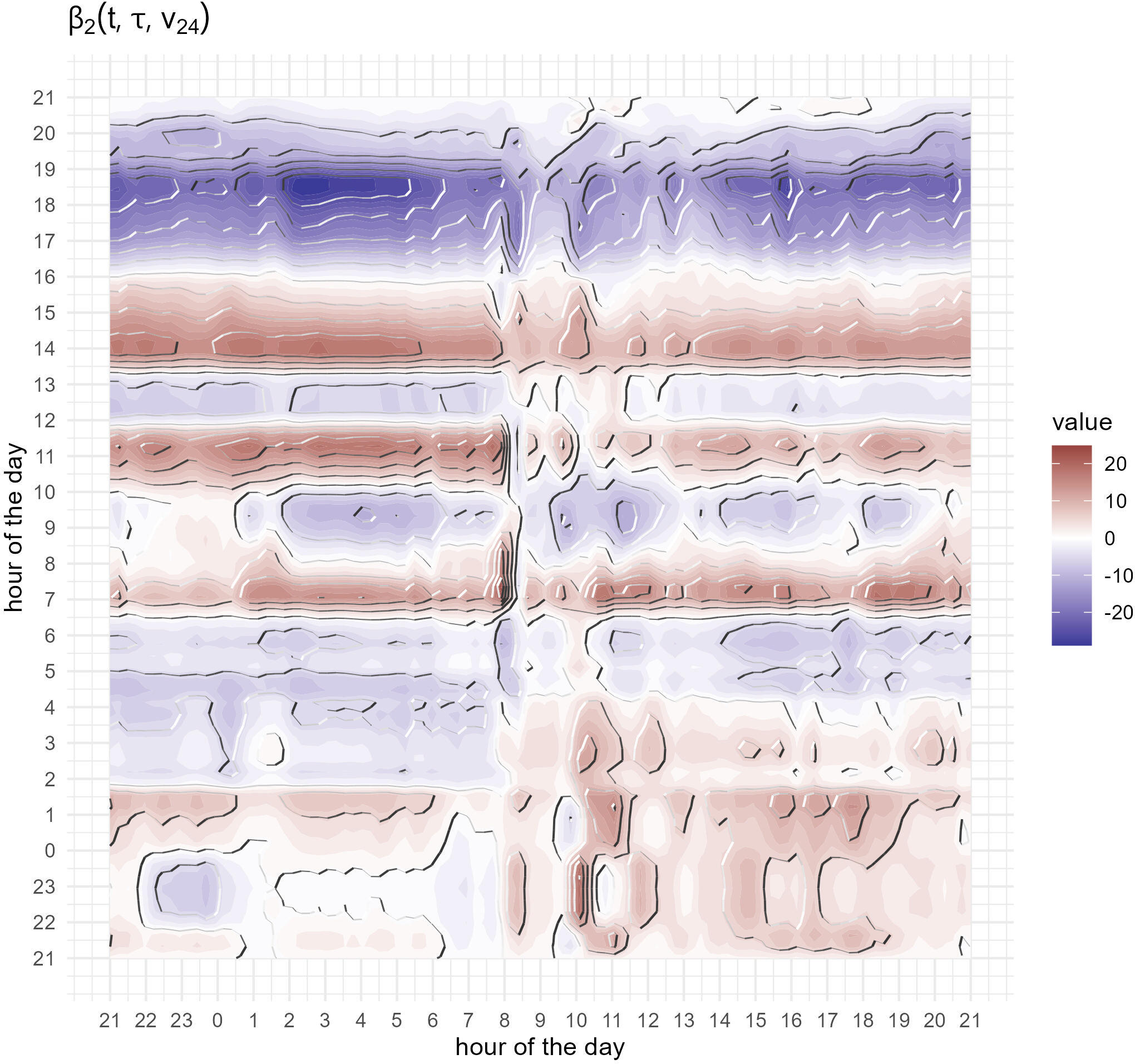}
		d)\includegraphics[scale=0.07]{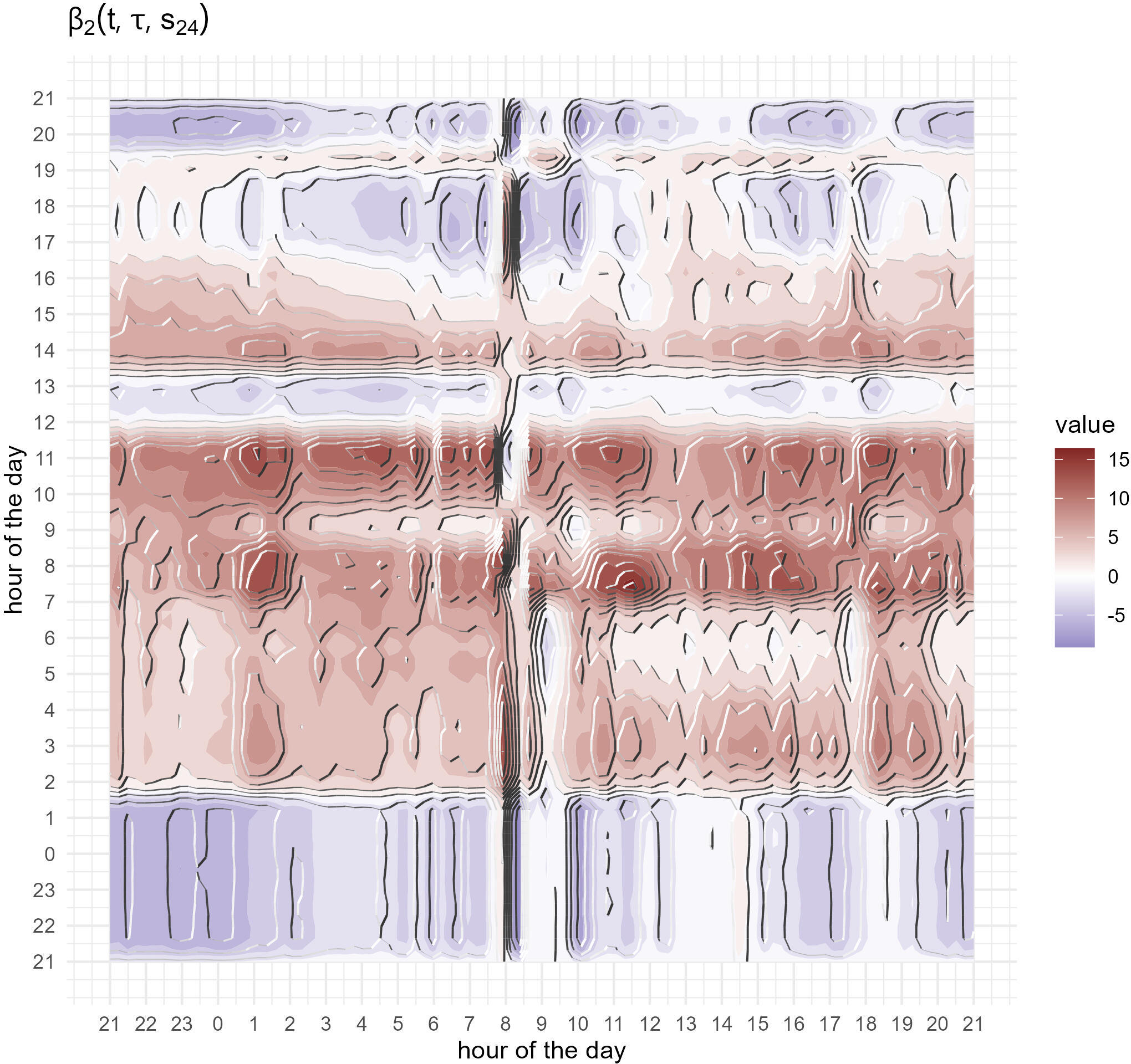}
		\caption{The regression coefficient functions $\beta_1(t, s, \nu_{24})$ (up) and $\beta_2(t, s, \nu_{24})$ (down) respectively for the NWFR (left) and GWFR (right).}
		\label{fig:02}
	\end{figure}		
	We observe that $TEMP$ throughout the day has different effects on $HUM$. Focusing on  $\beta_1(t, s, \nu_{24})$ obtained by NWFR, panel  $a)$ in Figure \ref{fig:02} , we note: a negative effect between $14:30$ and $16:00$, a moderate negative effect between the periods: $21:00-05:00$, $07:30-08:30$,  $09:30-11:00$ and finally $17:30-21:00$, a moderate positive effect between the periods: $05:00-07:30$,  $08:30 - 09:30$, $11:00 - 14:30$, and finally $16:00 - 17:30$. For instance, when $LUX$ is set between $14:30-16:00$, an increase of $+1^{\circ}C$ in $TEMP$ causes a decrease of $HUM$ ranging from $-40\%$ to  $-20\%$. However, if  $LUX$ is fixed between the intervals $05:00-07:30$,  $08:30 - 09:30$, $11:00 - 14:30$, or  $16:00 - 17:30$ an increase of $+1^{\circ}C$ results in an increase of $HUM$  between $0\%$ and $+20\%$. Also $LUX$, throughout the day, has different effects on $HUM$ and focusing on $\beta_2(t, s, \nu_{24})$  obtained by NWFR, panel  $c)$ in Figure \ref{fig:02}, we note: a strong negative effect between 16:00 and 21:00; a moderate negative effect between the periods: 02:00-06:30,  08:00-10:00 and finally 12:00-13:30; a moderate positive effect between 21:00 and 02:00; a positive and strong positive effect between the periods: 10:00-12:00, 13:30-16:00. When $TEMP$ is fixed between 16:00 and 21:00, a unitary $Lux$ increase results in a decrease in $HUM$ ranging from $-25\%$ to $-15\%$. In contrast, when $TEMP$ is fixed between 10:00-12:00,  13:30-16:00, a unitary $Lux$ increase leads to an increase in $HUM$ between $+10\%$ to $+20\%$.
	
	The accuracy of the classical functional regression model (namely, without considering the network structure), NWFR and GWFR models can be evaluated by considering the $RIMSE$, $R^2$ and $\widetilde{R}^2$ indices, which are shown in Tab. \ref{tab:lab_GOF}.  
	
	\begin{table}[H]
		\centering
		
		\label{tab:lab_GOF}
		\begin{tabular}{lrrr}
			\toprule
			Models & RIMSE & $R^2$ & $\widetilde{R}^2$ \\
			\midrule
			class. fr & 0.3439 & 0.986 & 0.987 \\
			GWFR      & 0.0018 & 1.000 & 1.000 \\
			NWFR      & 0.0013 & 1.000 & 1.000 \\
			\bottomrule
		\end{tabular}\caption{Lab. data: Model goodness-of-fit metrics}
	\end{table}
	
	Table~\ref{tab:lab_GOF} shows that the GWFR and NWFR models achieve near-perfect $R^2$ and  $\widetilde{R}^2$ values (equal to 1.000), along with very low RIMSE values. While these metrics indicate an excellent fit to the data, they may also suggest overfitting. To assess the generalizability of these models, we additionally conducted a permutation test.
	Table~\ref{tab:lab_pvals} summarizes the $p$-values of the regression coefficients of the NWFR and GWFR estimated models.
	\begin{table}[H]
		\centering
		\begin{tabular}{lrr}
			\toprule
			Models & $\beta_{TEMP}$ & $\beta_{LUX}$ \\
			\midrule
			GWFR & 0.012 & 0.292 \\
			NWFR & 0.129 & 0.408 \\
			\bottomrule
		\end{tabular}
		\caption{Lab. data: p-values of the geographical and network structure effect on the $\beta$ parameters.}
		\label{tab:lab_pvals}
	\end{table}
	Note that NWFR model can be considered as a special case of NWFR, in this case it seems that only $\beta_{TEMP}$ is significantly affected by the network structure, while the effect on the NWFR is sligtly moderate. $\beta_{LUX}$ variability seems to be less affected by the network topology. The predicted functional attributes for each vertex of the network can be evaluated through the split-conformal approach described above. We evaluate the empirical performances of the defined non-conformity measures $\mathcal{D}_{2}$ and $\mathcal{D}_\infty$, using $\alpha=0.05$. The evaluation includes empirical local and global coverage ( $Cov_{L}$ and  $Cov_G$), the average width of predicted intervals ( $ABW$), and the Interval Score ($S^{int}_{\alpha}$) as measures of efficiency.
	Results are summarized in Tab. \ref{tab:lab_pvals}.
	\begin{table}[H]
		\centering
		\begin{tabular}{llrrrr}
			\toprule
			Models & Conf. measure&$Cov_G\%$ & $Cov_L\%$& $ABW$& $S^{int}_{\alpha}$ \\
			\midrule
			Classical fr&$\mathcal{D}_{2}$ & 69.23 & 88.19 & 4.548 & 10.709 \\
			&$\mathcal{D}_\infty$& 82.69 & 93.10 & 6.158 & 9.150 \\
			\midrule
			GWFR &$\mathcal{D}_{2}$& 73.08 & 95.04 & 10.968 & 13.135 \\
			&$\mathcal{D}_\infty$& 92.31 & 99.67 & 16.890 & 16.960 \\
			\midrule
			NWFR &$\mathcal{D}_{2}$& 53.85 & 92.85 & 7.908 & 10.249 \\
			&$\mathcal{D}_\infty$& 94.23 & 99.67 & 12.374 & 12.491 \\
			\bottomrule
		\end{tabular}
		\caption{Lab.data: split-conformal prediction evaluation for the GWFR and NWFR models.}
		\label{tab:lab_pvals}
	\end{table}
	We observe that global coverage generally increases when employing the GWFR and NWFR models, except in the case of NWFR when the distance metric $\mathcal{D}_2$ is used. Examining the local coverage measure ($Cov_L$), we note that it generally improves when moving from the classical model to the GWFR and NWFR approaches, with a particularly pronounced improvement when $\mathcal{D}_\infty$ is employed.
	
	Regarding efficiency-related metrics, the use of $\mathcal{D}_2$ results in superior efficiency, as indicated by the generally lower values of $ABW$ and $S^{int}_\alpha$ compared to those obtained with $\mathcal{D}_\infty$. Given the inherent properties of $\mathcal{D}_2$ and $\mathcal{D}\infty$, these findings were anticipated.
	
	In summary, we conclude that while $\mathcal{D}_2$ yields greater efficiency relative to $\mathcal{D}_\infty$, it does so at the cost of slightly reduced coverage. This outcome aligns with the well-known \textit{coverage-efficiency} trade-off frequently discussed in the literature.

	\section{Concluding remarks} 
	\label{sec:07}
	
	In this work, we have introduced the Network-Weighted Functional Regression (NWFR) model, a novel extension of the classical functional regression framework that incorporates network structure through a weighting scheme based on the underlying graph topology. This approach enables the modeling of complex dependencies across networked observations, leading to improved predictive performance and enhanced interpretability in the analysis of functional data observed on networks.
	
	To provide uncertainty quantification, we developed a conformal prediction procedure tailored to the NWFR framework, delivering distribution-free prediction bands with guaranteed finite-sample marginal coverage. Through extensive simulations and an application to real-world environmental sensor data, we have demonstrated that explicitly accounting for the network structure yields substantial gains in predictive accuracy and in the validity of prediction intervals, relative to classical functional regression approaches.
	
	An important finding emerging from our analysis concerns the role of the conformity measure in balancing the validity-efficiency trade-off inherent in conformal prediction. While we considered commonly used nonconformity scores, the systematic study and optimization of conformity measures for functional data on networks remains an open problem of considerable practical and theoretical interest. To the best of our knowledge, no comprehensive study addressing this issue is currently available in the literature, and we identify this as a promising direction for future research.
	
	Further extensions of the NWFR framework could include the integration of additional network characteristics, such as node or edge pecific covariates, centrality measures, or temporal dynamics. The application of NWFR to other domains with complex network dependencies, such as social, biological, or transportation networks, also represents an important avenue for future investigation.
	
	%
	%
	
	\section*{Declarations}
	\paragraph{Conflict of interest} 
	The authors declare no conflict of interest.
	
	\bibliographystyle{abbrvnat}
	\bibliography{main}
	
\end{document}